\documentclass[a4paper,11pt]{article}
\usepackage{oldgerm}
\usepackage{color}
\usepackage{bbm}
\usepackage{amssymb}
\usepackage{epsfig}
\usepackage{amsthm}
\usepackage{amsmath, mathrsfs,psfrag}

\renewcommand{\mathbb}[1]{\mathbbm{#1}}

\newcommand{\Cl}{\mathbbm{C}}
\newcommand{\Rl}{\mathbb{R}}
\newcommand{\Nl}{\mathbb{N}}

\definecolor{lightgray}{rgb}{0.8,0.8,0.8}

\newcommand{\Om}{\Omega}
\newcommand{\om}{\omega}

\newcommand{\te}{\theta}

\newcommand{\la}{\lambda}

\newcommand{\eps}{\varepsilon}

\newcommand{\R}{\mathcal{R}}
\newcommand{\B}{\mathcal{B}}

\newcommand{\M}{\mathcal{M}}

\newcommand{\NN}{\mathcal{N}}

\newcommand{\Hil}{\mathcal{H}}

\newcommand{\U}{\mathcal{U}}
\newcommand{\DD}{\mathcal{D}}

\newcommand{\Ss}{\mathscr{S}}   

\newcommand{\fti}{\tilde{f}}

\newcommand{\PG}{\mathcal{P}}

\newcommand{\frS}{\textfrak{S}}

\newcommand{\supp}{\text{supp}\,}

\newcommand{\Strip}{\mathrm{S}}

\DeclareMathOperator{\im}{Im}

\newcommand{\ot}{\otimes}

\newcommand{\ad}{a^{\dagger}}

\usepackage{amsmath}
\usepackage{amsfonts}
\usepackage{amssymb}
\usepackage{mathrsfs}
\usepackage[pdfborder={0 0 0}]{hyperref}
\usepackage{color}
\usepackage[T1]{fontenc}
\usepackage{lmodern}

\newtheorem{theorem}{Theorem}[section]
\newtheorem{proposition}[theorem]{Proposition}
\newtheorem{lemma}[theorem]{Lemma}

\newtheorem{definition}[theorem]{Definition}

\numberwithin{equation}{section}

\newlength{\dinwidth}
\newlength{\dinmargin}

\usepackage{paralist}
\newenvironment{definitionlist}{\begin{compactenum}[\itshape i)]}{\end{compactenum}}
\newcommand{\refitem}[1] {\textit{\ref{#1})}}

\DeclareMathOperator{\Symm}{Symm}

\begin{document}

\title{\bf On the Equivalence of Two Deformation Schemes in Quantum Field Theory}
\author{
	Gandalf Lechner
	\footnote{Institute for Theoretical Physics, University of Leipzig, 04104 Leipzig, Germany, email: \texttt{gandalf.lechner@uni-leipzig.de}}
	\qquad
	Jan Schlemmer
	\footnote{Department of Physics, University of Vienna, 1090 Vienna, Austria, email: \texttt{jan.schlemmer@univie.ac.at}
	GL and JS supported by FWF project P22929-N16 ``Deformations of quantum field theories''}
	\qquad
	Yoh Tanimoto
	\footnote{Institute for Theoretical Physics, University of G\"ottingen, 37077 G\"ottingen, Germany, email: \texttt{yoh.tanimoto@theorie.physik.uni-goettingen.de}.	Supported by Deutscher Akademischer Austauschdienst.
	}
	}
\date{September 11, 2012}

\maketitle

\begin{abstract}
Two recent deformation schemes for quantum field theories on two-dimensional Minkowski space, making use of deformed field operators and Longo-Witten endomorphisms, respectively, are shown to be equivalent.
\end{abstract}

\noindent{\small {\bf Keywords:} deformations of quantum field theories, two-dimensional models, modular theory}


\section{Deformations of QFTs by inner functions and their roots}

In recent years, there has been a lot of interest in deformations of quantum field theories \cite{GrosseLechner:2007,BuchholzSummers:2008,BuchholzLechnerSummers:2010,DybalskiTanimoto:2010,LongoWitten:2010,LongoRehren:2011, Morfa-Morales:2011,Lechner:2011, Tanimoto:2011-1,BischoffTanimoto:2011, Alazzawi:2012,Much:2012, Plaschke:2012} in the sense of specific procedures modifying quantum field theoretic models on Minkowski space, mostly motivated by the desire to construct new models in a non-perturbative manner. Various constructions have been invented, relying on different methods such as smooth group actions, non-commutative geometry, chiral conformal field theory, boundary quantum field theory, and inverse scattering theory.

In many situations, it is possible to set up the deformation in such a way that Poincar\'e covariance is completely preserved and locality partly. More precisely, often the deformation introduces operators which are no longer localized in arbitrarily small regions of spacetime, but rather in unbounded regions like a Rindler wedge $W:=\{x\in\Rl^d\,:\,x_1>|x_0|\}$. In the operator-algebraic framework of quantum field theory \cite{Haag:1996}, such a wedge-local Poincar\'e covariant model can be conveniently described by a so-called Borchers triple $(\M,U,\Om)$ \cite{Borchers:1992,BuchholzLechnerSummers:2010}, consisting of a von Neumann algebra $\M$ of operators localized in the wedge $W$, a suitable representation $U$ of the translations, and an invariant (vacuum) vector $\Om$ (see Def.~\ref{definition:WedgeAlgebra} below). Depending on the method at hand, the von Neumann algebra $\M$ is generated by different objects, like deformed field operators or twisted chiral observables.

It is the aim of this letter to show that some of the constructions on two-dimensional Minkowski space are identical in the sense of unitary equivalence of their associated Borchers triples. More precisely, we will show that the deformations presented in \cite{Tanimoto:2011-1}, starting from a chiral field theory, are equivalent to the deformations in terms of deformed field operators, presented in \cite{Lechner:2011}, for mass $m=0$ and dimension $d=1+1$ (Section~\ref{section:m=0}). In the special case of the so-called warped convolution deformation \cite{BuchholzLechnerSummers:2010}, such an equivalence was already observed in \cite{Tanimoto:2011-1}. Here we prove that also for the infinite family of deformations considered in \cite{Lechner:2011}, one obtains the same construction as in \cite{Tanimoto:2011-1} in the chiral situation, where the deformation amounts to a unitary equivalence transformation by a Longo-Witten endomorphism \cite{LongoWitten:2010} on each light ray. Furthermore, we will show that
certain aspects of the chiral construction carry over to the massive situation (Section~\ref{section:m>0}).

The deformations we are interested in here take certain families of analytic functions as input parameters, whose relations we will now clarify. We will write $\mathbb{H}\subset\Cl$ for the open upper half plane, $\Strip(0,\pi):=\{\zeta\in\Cl\,:\,0<\im\zeta<\pi\}$ for the strip, and $H^\infty(\mathbb H)$, $H^\infty(\Strip(0,\pi))$ for the Hardy spaces of bounded analytic functions on these domains. Recall that for a function $f\in H^\infty(\mathbb H)$, the limit $\lim_{\eps\searrow 0}f(t+i\eps)$ exists almost everywhere\footnote{By ``almost everywhere'' (a.e.) and ``almost all'' we always refer to Lebesgue measure on $\Rl$.} and defines a boundary value function in $L^\infty(\Rl)$. The same holds for functions in $H^\infty(\Strip(0,\pi))$ and their boundary values at $\Rl$ and $\Rl+i\pi$.

\begin{definition}\label{definition:FunctionFamilies}
	\begin{definitionlist}
		\item\label{item:SymmetricInnerFunctions} A symmetric inner function is a function $\varphi\in H^\infty(\mathbb H)$ whose boundary values on the real line satisfy $\overline{\varphi(t)}=\varphi(t)^{-1}=\varphi(-t)$ for almost all $t\in\Rl$.
		\item\label{item:Roots} A root of a symmetric inner function $\varphi$ is a function $R\in L^\infty(\Rl)$ such that $\overline{R(t)}=R(t)^{-1}=R(-t)$ and $R(t)^2=\varphi(t)$ for almost all $t\in\Rl$. The family of all roots of symmetric inner functions will be denoted $\R$.
		\item\label{item:ScatteringFunctions} A scattering function is a function $S\in H^\infty(\Strip(0,\pi))$ whose boundary values satisfy $\overline{S(\theta)}=S(\theta)^{-1}=S(-\theta)=S(i\pi+\theta)$ for almost all $\te\in\Rl$.
	\end{definitionlist}
\end{definition}

Symmetric inner functions provide the input into deformations making use of Longo-Witten endomorphisms \cite{LongoWitten:2010,Tanimoto:2011-1, LongoRehren:2011, BischoffTanimoto:2011}, whereas scattering functions are used in inverse scattering approaches such as \cite{Lechner:2003, BostelmannLechnerMorsella:2011}. For convenience, the latter are usually defined with the additional requirement of extending continuously to the closure of $\Strip(0,\pi)$. However, going through the construction, say in \cite{Lechner:2003}, one realizes that this continuity assumption is not necessary. What is required is that the boundary conditions on $S$ hold almost everywhere, the boundary values are regular enough to define multiplication operators on $L^2(\Rl)$, and for $f\in H^\infty(\Strip(0,\pi))$ with Schwartz boundary values, $[0,\pi]\ni\la\mapsto \int_\Rl d\te\,f(\te+i\la)S(\te+i\la)$ is continuous. As this is the case for any $S\in H^\infty(\Strip(0,\pi))$, one can just as well work with the more general definition
of scattering function given above.

We also note that scattering functions and symmetric inner functions are in one to one correspondence by $S(\zeta):=\varphi(\sinh\zeta)$, $\zeta\in \Strip(0,\pi)$. As $\sinh(i\pi+\zeta)=-\sinh\zeta=\sinh(-\zeta)$, this identification produces the required properties of the boundary values in Def.~\ref{definition:FunctionFamilies}~\refitem{item:ScatteringFunctions}. On the other hand, $\varphi(z):=S(\sinh^{-1}z)$, $z\in\mathbb H$, is well-defined and analytic because of the crossing symmetry $S(i\pi+\te)=S(-\te)$ of $S$. This identification of the strip and the half plane via $\sinh$ is the one encountered in massive theories \cite{GrosseLechner:2007}. In massless theories, also the identification $\exp:\Strip(0,\pi)\to\mathbb H$ occurs \cite{LongoWitten:2010}, and under this identification, scattering functions correspond to the subset of symmetric inner functions with the additional symmetry $\overline{\varphi(t)}=\varphi(t^{-1})$ , see \eqref{eq:R0-generalized} below.

Regarding Def.~\ref{definition:FunctionFamilies}~\refitem{item:Roots}, we note that each symmetric inner function has infinitely many different roots, and the family of roots $\R$ contains all symmetric inner functions because Def.~\ref{definition:FunctionFamilies}~\refitem{item:SymmetricInnerFunctions} is stable under taking squares. These roots are the input into the deformations in \cite{Lechner:2011,Alazzawi:2012,Plaschke:2012}, where under additional regularity assumptions, they are called {\em deformation functions}. We note that these additional requirements are only necessary when working on the tensor algebra of test functions \cite{Lechner:2011}, but not when working directly on a representation space such as in \cite{Alazzawi:2012}. In particular, the roots will not be required to be analytic, and also the condition $R(0)=1$ \cite{Lechner:2011}, related to fixing a root of an inner function, will
not be assumed here.
\\
\\
In the following, we will be concerned with deformations of free field theories of mass $m\geq0$ on two-dimensional Minkowski space, and now set up some standard notation for this. We will be working on the Bose Fock space
\begin{align*}
	\Hil
	:=
	\Gamma(\Hil_1)
	\,,\qquad
	\Hil_1
	:=
	L^2(\Rl,\tfrac{dp}{\om_m(p)})
	\,,
	\quad
	\om_m(p):=(m^2+p^2)^{1/2}
	\,.
\end{align*}
Its Fock vacuum will be denoted $\Om$, and we have the usual representation $\Gamma(U_1)$ of the proper Poincar\'e group as the second quantization of
\begin{align}\label{eq:U1}
	[U_1(x,\la)\Psi_1](p)
	&:=
	e^{i(x_0\om_m(p)-x_1p)}
	\,\Psi_1(\la p)
	\,,\qquad
	[U_1(j)\Psi_1](p)
	:=
	\overline{\Psi_1(p)}
	\,,
\end{align}
where $x=(x_0,x_1)\in\Rl^2$ is the translation, $\la\in\Rl$ denotes the boost rapidity parameter, $\la p:=-\sinh\la\cdot \om_m(p)+\cosh\la\cdot p$, and $j(x)=-x$ is the space-time reflection. We will also write $U(x):=\Gamma(U_1(x,0))$ for the translations.

From an operator-algebraic point of view, a wedge-local quantum field theory is equivalent to a Borchers triple.

\begin{definition}\label{definition:WedgeAlgebra}
	A Borchers triple $(\M,U,\Om)$ on $\Rl^2$ consists of a von Neumann algebra $\M\subset\B(\Hil)$, a strongly continuous unitary positive energy representation $U$ of the translation group $\Rl^2$ on $\Hil$, and a $U$-invariant unit vector $\Om\in\Hil$ such that
	\begin{definitionlist}
		\item $U(x)\M U(x)^{-1}\subset\M$ for any $x\in\overline{W}$,
		\item $\Om$ is cyclic and separating for $\M$.
	\end{definitionlist}
	Two Borchers triples $(\M,U,\Om)$ and $(\tilde{\M},\tilde{U},\tilde{\Om})$ will be called equivalent, written $(\M,U,\Om)\cong(\tilde{\M},\tilde{U},\tilde{\Om})$, if there exists a unitary $V$ such that $V\M V^*=\tilde{\M}$, $VU(x)V^*=\tilde{U}(x)$ for all $x\in\Rl^2$, and $V\Om=\tilde{\Om}$.
\end{definition}

Recall that by a famous theorem of Borchers \cite{Borchers:1992}, the representation $U$ can be extended to a (anti-) unitary representation $U_\M$ of the proper Poincar\'e group $\PG_+$ with the help of the modular data $J_\M,\Delta_\M$ of $(\M,\Om)$, by
\begin{align}
	U_\M(x,\la)
	:=
	U(x)\Delta_\M^{-\frac{i\la}{2\pi}}
	\,,\qquad
	U_\M(j)
	:=
	J_\M
	\,.
\end{align}
As is well known, a Borchers triple gives rise to a Poincar\'e-covariant net of wedge algebras \cite{Borchers:1992}, which can under further conditions be extended to a net of double cone algebras \cite{BuchholzLechner:2004,Lechner:2008}. We will not discuss the extension question here, but rather focus on the wedge-local aspects only. Note that the net of wedge algebras generated from a Borchers triple $(\M,U,\Om)$ will transform covariantly under a representation $U_\M$ of the Poincar\'e group which depends on $\M$. However, in the case of two equivalent Borchers triples $(\M,U,\Om)\cong(\tilde{\M},\tilde{U},\tilde{\Om})$, modular theory tells us that the modular data of $(\M,\Om)$ and $(\tilde{\M},\tilde{\Om})$ are related by $VJ_\M V^*=J_{\tilde{\M}}$, $V\Delta_\M^{it}V^*=\Delta_{\tilde{\M}}^{it}$, i.e. equivalence of Borchers triples implies equivalence of the associated wedge-local nets including their representations $U_\M\cong U_{\tilde{\M}}$ of the proper Lorentz group.
\\
\\
A particular example of a Borchers triple is provided by the model of a  free scalar quantum field: Let $a(\xi)$ and $\ad(\xi):=a(\xi)^*$, $\xi\in\Hil_1$, denote the standard CCR annihilation and creation operators on $\Hil$, and for $f\in\Ss(\Rl^2)$, let
\begin{align}
	\phi_m(f)
	&:=
	\ad(f^+)+a(\overline{f^-})
	\,,\qquad
	f^\pm(p)
	:=
	\fti(\pm\om_m(p),\pm p)\,,
\end{align}
denote the free Klein-Gordon field of mass $m\geq0$ (with $f$ restricted to directional derivatives of test functions in the case $m=0$ because of the well-known infrared divergence in the measure $\frac{dp}{|p|}$). With the wedge algebra
\begin{align}
	\M_m
	:=
	\{e^{i\phi_m(f)}\,:\,f\in\Ss_\Rl(W)\}''
	\,,
\end{align}
the Fock translations $U$ \eqref{eq:U1} and the Fock vacuum $\Om$, we then have a Borchers triple $(\M_m,U,\Om)$. In this case, the modular data of $(\M_m,\Om)$ reproduce the Poincar\'e representation \eqref{eq:U1}, i.e. $U_{\M_m}=\Gamma(U_1)$. For convenience of notation, we will write
\begin{align}
	J:=J_{\M_m}=\Gamma(U_1(j))
	\,,\qquad
	\Delta^{it}:=\Delta^{it}_{\M_m}=\Gamma(U_1(0,-2\pi t))
	\,.
\end{align}

Fixing the representation $U$ of the translations and the vector $\Om$, the algebra $\M_m$ is however by no means the only von Neumann algebra completing $U,\Om$ to a Borchers triple. In the following, we will introduce for each $R\in\R$ two von Neumann algebras $\M_{R,m}$, $\NN_R$ with this property, obtained by (generalizations of) the deformation procedures in \cite{Lechner:2011} and \cite{Tanimoto:2011-1}, respectively. For $R=1$, both families reduce to the undeformed situation, i.e. $\M_{1,m}=\M_m$, $\NN_{1}=\M_0$.

To define the first set of deformed wedge algebras $\M_{R,m}$, we introduce a unitary-valued function $T_{R,m}:\Rl\to\U(\Hil)$ \cite{Lechner:2011}
\begin{align}
	[T_{R,m}(p)\Psi]_n(p_1,\ldots,p_n)
	&:=
	\prod_{k=1}^n
	R_m(p,p_k)
	\cdot
	\Psi_n(p_1,\ldots,p_n)
	\,.
\end{align}
Here the function $R_m\in L^\infty(\Rl^2)$ is in the case of positive mass defined as
\begin{align}\label{eq:R_m}
	R_m(p,q)
	&:=
	R\big(\tfrac{1}{2}(\om_m(q)p-\om_m(p)q)\big)
	\,,\qquad
	m>0\,,
\end{align}
where the factor $\frac{1}{2}$ is a matter of convention. Taking the limit $m\to0$, one observes that the argument $\frac{1}{2}(|q|p-|p|q)$ of $R$ vanishes if $p$ and $q$ have the same sign. As the root $R$ is only defined up to equivalence in $L^\infty(\Rl)$, its value at $0$ is not fixed. We therefore define
\begin{align}\label{eq:R0}
	R_0(p,q)
	:=
	\left\{
		\begin{array}{ccl}
			R(-pq) &;& p>0,\,q<0\\
			R(+pq) &;& p<0,\,q>0\\
			1 &;& p>0,\,q>0\;\;\text{or}\;\; p<0,\,q<0
		\end{array}
	\right.
	\;.
\end{align}
Note that for any mass $m\geq0$, we have for almost all $p,q\in\Rl$
\begin{align}
	R_m(q,p)=R_m(p,q)^{-1}
	\,,\qquad
	m\geq0\,.
\end{align}

The assignment $a_R(p):=a(p)T_{R,m}(p)$ defines an operator-valued distribution for any $R\in\R$, which explicitly acts on a vector $\Psi$ of finite particle number according to, $\xi\in\Hil_1$,
\begin{align}\label{eq:a_R-explicit}
	[a_R(\xi)\Psi]_n(p_1,\ldots,p_n)
	=
	\sqrt{n+1}\int \frac{dq}{\om_m(q)}\,\overline{\xi(q)}\,\prod_{k=1}^n R_m(q,p_k)\Psi_{n+1}(q,p_1,\ldots,p_n)
	\,.
\end{align}
Its adjoint is denoted $\ad_R(\xi):=a_R(\overline{\xi})^*$, and the corresponding deformed field operator is
\begin{align}
	\phi_{R,m}(f)
	:=
	\ad_R(f^+)+a_R(\overline{f^-})
	\,,
	\qquad
	f\in\Ss(\Rl^2)\,.
\end{align}
As $\phi_{R,m}(f)$ is essentially self-adjoint on the subspace of finite particle number for real $f$, one can pass to the generated von Neumann algebra
\begin{align}\label{eq:MRm}
	\M_{R,m}
	:=
	\{e^{i\phi_{R,m}(f)}\,:\,f\in\Ss_\Rl(W)\}''
	\,.
\end{align}
\begin{theorem}\label{theorem:M}
	Let $R\in\R$ and $m\geq0$. Then $(\M_{R,m},U,\Om)$ is a Borchers triple with modular data $J_{\M_{R,m}}=J$ and  $\Delta^{it}_{\M_{R,m}}=\Delta^{it}$.
\end{theorem}
For $m>0$, this has been established in \cite{Lechner:2011}, and for $m=0$, one can use essentially the same proofs, so that we do not have to go into details here. In fact, as for $m=0$ the mass shell decomposes into two half-rays which are left invariant by the Lorentz boosts, one can in this case more generally consider {\em three} roots $R,R_1,R_2\in\R$, with the additional requirement $\overline{R_k(t)}=R_k(t^{-1})$ for almost all $t\in\Rl$, $k=1,2$, and put
\begin{align}\label{eq:R0-generalized}
	R_0(p,q)
	:=
	\left\{
		\begin{array}{ccl}
			R(-pq) &;& p>0,\,q<0\\
			R(+pq) &;& p<0,\,q>0\\
			R_1(+\tfrac{p}{q}) &;& p>0,\,q>0\\
			R_2(-\tfrac{p}{q}) &;& p<0,\,q<0
		\end{array}
	\right.
	\;.
\end{align}
Also with this more general definition of $R_0$, the algebra $\M_{R,0}$ completes $U,\Om$ to a Borchers triple. In the terminology of \cite{FendleySaleur:1993}, the functions $R_1$, $R_2$ govern the left-left and right-right ``scattering'' of the model, whereas $R$ determines the left-right (wave) scattering \cite{Bischoff:2012, Buchholz:1977, DybalskiTanimoto:2010}. If $R=1$, the corresponding model is chiral -- this is in particular the case for the short distance scaling limits of the models generated by the massive wedge algebras $\M_{R,m}$, $m>0$. In this context one finds $R=1$, $R_1(t)^2=R_2(t)^2=\varphi(t-t^{-1})$ with some symmetric inner function $\varphi$ \cite{BostelmannLechnerMorsella:2011}. For the purposes of this letter, we will however restrict ourselves to the case $R_1=R_2=1$ \eqref{eq:R0}, which corresponds to the construction in \cite{Tanimoto:2011-1}.
\\
\\
To define the second set of deformed wedge algebras $\NN_R$, one works in the massless case $m=0$, and uses the chiral structure present in this situation. Here the Fock space, the representation $U$, the invariant vector $\Om$, and the wedge algebra $\M_0$ split into two (chiral) factors. With $\Hil_1^\pm:= L^2(\Rl_\pm,\frac{dp}{|p|})$, we have $\Hil_1=\Hil_1^+\oplus\Hil_1^-$ and
\begin{align}
	\Hil
	&\cong
	\Hil^+
	\ot
	\Hil^-
	\,,\qquad
	\Hil^\pm:=\Gamma(\Hil_1^\pm)
	\,,\\
	\Om&\cong\Om_+\ot\Om_-
	\,,
	\\
	U(x)&\cong U_+(x_-)\ot U_-(x_+)
	\,,\qquad
	x_\pm:=x_0\pm x_1
	\,,\\
	\M_0
	&\cong
	\M_{0,+}\ot\M_{0,-}\,,
\end{align}
where $\Om_\pm$ denotes the Fock vacuum in $\Hil^\pm$. The canonical unitary $V:\Hil^+\ot\Hil^-\to\Hil$ realizing the above isomorphisms is recalled in \eqref{eqn_isoexpo}. Note that
\begin{align}
	V^*\Gamma(U_1(x,\la))V
	&=
	\Gamma_+(U_{1,+}(x_-,\la))\ot\Gamma_-(U_{1,-}(x_+,\la))
	\,,\\
	V^*\Gamma(U_1(j))V
	&=
	\Gamma_+(U_{1,+}(j))\ot\Gamma_-(U_{1,-}(j))
	\,,
\end{align}
where $\Gamma_\pm$ denotes second quantization on $\Hil^\pm$, with $(U_{1,\pm}(x_\mp,\la)\Psi_1)(p)=e^{\pm ipx_\mp}\Psi_1(e^{\mp\la}\cdot p)$ and $(U_{1,\pm}(j)\Psi_1)(p)=\overline{\Psi_1(p)}$. For the sake of a concise notation, we will write
\begin{align}
	J_\ot
	&:=
	\Gamma_+(U_{1,+}(j))\ot\Gamma_-(U_{1,-}(j))
	=
	V^*JV
	\,,\\
	\Delta_\ot^{it}
	&:=
	\Gamma_+(U_{1,+}(0,-2\pi t))\ot\Gamma_-(U_{1,-}(0,-2\pi t))
	=
	V^*\Delta^{it}V
	\,.
\end{align}
Given $R\in\R$, one introduces the unitary $S_R\in\U(\Hil^+\ot\Hil^-)$ \cite{Tanimoto:2011-1},
\begin{align}
	\big[ S_R \Psi \big]_{n,n'}(p_1, \ldots, p_n, q_1, \ldots, q_{n'})
	=
	\prod_{\substack{i=1\ldots n\\j=1\ldots n'}} R_0(p_i,q_j) \cdot
	\Psi_{n,n'}(p_1, \ldots, p_n, q_1, \ldots, q_{n'})
	\label{eqn_Saction}
	\,,
\end{align}
and defines the von Neumann algebra
\begin{align}\label{eq:NR0}
	\NN_R
	:=
	(\M_{0,+}\ot1)\vee S_{R^2}(1\ot\M_{0,-})S_{R^2}^*
	\,.
\end{align}
\begin{theorem}\label{theorem:N}
	Let $R\in\R$. Then $(\NN_R,U_+\ot U_-,\Om_+\ot \Om_-)$ is a Borchers triple with modular data $J_{\NN_R}=S_{R^2}J_\ot$ and  $\Delta^{it}_{\NN_R}=\Delta_\ot^{it}$.
\end{theorem}

This theorem has been proven in \cite{Tanimoto:2011-1}. From \eqref{eq:NR0}, it is clear that $\NN_R$ depends on $R$ only via the symmetric inner function $R^2$.

Our results can now compactly be summarized as follows ($R\in\R$, $m\geq0$):
\begin{itemize}
	\item $(\NN_R,U_+\ot U_-,\Om_+\ot\Om_-)\cong(\M_{R,0},U,\Om)$ (Theorem~\ref{thm:Algeq}).
	\item $(\M_{R_1,m},U,\Om)\cong(\M_{R_2,m},U,\Om)$ if and only if $R_1^2=R_2^2$ (Proposition~\ref{proposition:ChoiceOfRootDoesntMatter}).
\end{itemize}
As $R^2$ is essentially the two-particle S-matrix of the model described by the Borchers triple $(\M_{R,m},U,\Om)$, the last result amounts to a proof of uniqueness of the solution of the inverse scattering problem in the setting of the deformations studied here. In case of continuous $R$, such an effect was already observed in \cite{Alazzawi:2012}.
For massless nets obeying a number of natural conditions, uniqueness of the inverse scattering problem is known
once one fixes the asymptotic algebra \cite{Tanimoto:2011-1}. Furthermore, explicit examples of
S-matrices not preserving the Fock space structure are known in the massless case \cite{BischoffTanimoto:2011}.

The convenient deformation formula \eqref{eq:NR0} is a result of the chiral structure present in the massless case, and has no direct analogue in the massive case. In Section~\ref{section:m>0}, we will discuss why the situation is more complex in the massive case even though (formal) relations between deformed and
undeformed creation and annihilation operators still exist. 


\section{Equivalence of the two deformations in the massless case}\label{section:m=0}

The aim of this section is to demonstrate the equivalence of the mass zero Borchers triples $(\M_{R,0},U,\Om)$ \eqref{eq:MRm} and $(\NN_R,U_+\ot U_-,\Om_+\ot U_-)$ \eqref{eq:NR0} for arbitrary roots $R\in\R$. From Theorem~\ref{theorem:M} and Theorem~\ref{theorem:N}, we see that the modular groups of these von Neumann algebras coincide with the one parameter boost groups on their respective Hilbert spaces, but their modular conjugations differ by $S_{R^2}$, i.e. we have $J_{\NN_R}=S_{R^2}V^*J_{\M_{R,0}}V$ with the canonical unitary $V:\Hil^+\ot\Hil^-\to\Hil$ \eqref{eqn_isoexpo}. We will therefore in a first step go over to an equivalent form of $\NN_R$ which has modular conjugation $V^*J_{\M_{R,0}}V$, without the factor $S_{R^2}$. In general, this can be accomplished by conjugating with a root of the ``S-matrix'' $J_{\NN_R}V^*J_{\M_{R,0}}V$ \cite{Wollenberg:1992}, and in our present situation, this amounts to considering
\begin{align}\label{eq:Nhat0}
         \hat{\NN}_R
         :=
         S_R^* (\M_{0,+}\ot1)S_R\vee S_R(1\ot\M_{0,-})S_R^*
         \,.
\end{align}
\begin{lemma}
	Let $R\in\R$. Then $(\hat{\NN}_R,U_+\ot U_-,\Om_+\ot\Om_-)$ is a Borchers triple equivalent to $(\NN_R,U_+\ot U_-,\Om_+\ot\Om_-)$, with modular data
	\begin{align}\label{eq:Nhat-ModularData}
		J_{\hat{\NN}_R}
		=
		J_\ot
		\,,\qquad
		\Delta^{it}_{\hat{\NN}_R}
		=
		\Delta^{it}_\ot
		.
	\end{align}
\end{lemma}
\begin{proof}
	As $S_R$ \eqref{eqn_Saction} satisfies $S_R^2=S_{R^2}$, we have the unitary equivalence of algebras $\hat{\NN}_{R}=S_R^*\NN_RS_R$. The unitary $S_R$ clearly commutes with all translations $U_+(x_-)\ot U_-(x_+)$ and leaves $\Om_+\ot\Om_-$ invariant. Hence $(\hat{\NN}_R,U_+\ot U_-,\Om_+\ot\Om_-)\cong(\NN_R,U_+\ot U_-,\Om_+\ot\Om_-)$; this also shows that $(\hat{\NN}_R,U_+\ot U_-,\Om_+\ot\Om_-)$ is a Borchers triple.

	Regarding the modular data of $(\hat{\NN}_R,\Om_+\ot\Om_- )$, we first have by modular theory
	\begin{align}
		J_{\hat{\NN}_R}
		=
		S_R^*J_{\NN_R} S_R
		\,,\qquad
		\Delta_{\hat{\NN}_R}^{it}
		=
		S_R^*\Delta_{\NN_R}^{it}S_R
		=
		S_R^*\Delta_\ot^{it}S_R
	\,.
	\end{align}
	Taking into account the action of $\Gamma_\pm(U_{1,\pm}(0,\la))$, one sees from \eqref{eq:R0} and \eqref{eqn_Saction} that $S_R$ commutes with the Lorentz boosts. As these coincide with the modular unitaries $\Delta^{it}_\ot$, the second equation in \eqref{eq:Nhat-ModularData} follows. To establish the claimed form of the modular conjugation, we note $J_\ot S_R=S_{\overline{R}}J_\ot=S_R^*J_\ot$ and compute
	\begin{align*}
		J_{\hat{\NN}_R}
		&=
		S_R^*J_{\NN_R} S_R
		=
		S_R^*S_{R^2}J_\ot S_R
		=
		S_R^*S_R^2 S_R^*J_\ot
		=
		J_\ot
		\,.
	\end{align*}
	This completes the proof.
\end{proof}

The equivalence between the two deformed Borchers triples with wedge algebras $\M_{R,0}$ and $\hat{\NN}_R$ will now be established
using the creation and annihilation operators into which the fields
generating $\M_0$ can be decomposed. Corresponding to the splitting
$\M_0 = \M_{0,+} \otimes \M_{0,-}$ we have creation and annihilation
operators $a_\pm$, $a_\pm^\dagger$ acting on $\Hil^\pm$.

In the following, we will always suppress the canonical embeddings
$\iota^n_\pm:(\Hil^\pm_1)^{\otimes n}\to\Hil_1^{\otimes n}$,
$(\iota^n_\pm\Psi_n^\pm)(p_1,...,p_n):=\Psi_n^\pm(p_1,...,p_n)$
for $p_1,...,p_n\in\Rl_\pm$ and $(\iota^n_\pm\Psi_n^\pm)(p_1,...,p_n):=0$ otherwise.
With these embeddings understood, $\Hil_1 = \Hil_1^+ \oplus \Hil_1^-$, and hence
$\Hil \cong \Hil^+ \otimes \Hil^-$. This isomorphism is given explicitly by a
unitary
\[
	V: \Hil^+ \otimes \Hil^- \to \Hil\,,
\]
which is uniquely determined by its action on the total set \cite{Guichardet:1972} of ``exponential vectors'' $e^{\Psi_1}:=\sum_{n=0}^\infty\frac{1}{\sqrt{n!}}\Psi_1^{\otimes n}$ by
\begin{equation}
	V (e^{\Psi_1} \otimes e^{\Phi_1})
	:=
	e^{\Psi_1 \oplus \Phi_1}\,,
	\qquad
	\Psi_1\in \Hil_1^+,\;\Phi_1\in\Hil_1^-\,.
	\label{eqn_isoexpo}
\end{equation}

The definitions of $\M_{R,0}$ and $\NN_R$ make use of the realizations of $\Hil$ as $\Gamma(\Hil_1^+ \oplus \Hil_1^-)$ and $\Gamma(\Hil_1^+) \otimes \Gamma(\Hil_1^-)$, respectively. We will work on $\Hil=\Gamma(\Hil_1^+ \oplus \Hil_1^-)$, and first compute an explicit expression of $S_R$ on this space.

\begin{lemma}
Let $R\in\R$ and $\hat{S}_R:= V S_R V^*$. Then, $\Psi \in \Hil$,
\begin{align}\label{eq:Shat}
	\big[\hat{S}_R \Psi \big]_n (p_1, \ldots, p_n)
	=
	\prod_{i,j=1}^n
	R_0^+(p_i, p_j)
	\cdot
     \Psi_n(p_1, \ldots, p_n)
	\,,
\end{align}
where
\[
	R_0^+(p,q)
	=
	\begin{cases}
		R(- p q) & \ p>0, q<0 \\
		1     & \text{ else }
	\end{cases}
	\,.
\]
\end{lemma}

\begin{proof}
As exponential vectors form a total set in $\Hil$, it is sufficient to compute $\hat{S}_R$ on $e^{\Psi_1\oplus \Phi_1}$ to verify \eqref{eq:Shat}. The action of $V$ from \eqref{eqn_isoexpo} on vectors $\Xi=\sum_{n,m=0}^\infty\Xi_{n,m}\in\bigoplus_{n,m=0}^\infty((\Hil_1^+)^{\otimes_s n}\ot(\Hil_1^-)^{\otimes_s m})=\Hil^+\ot\Hil^-$
is explicitly given by
\begin{equation}
	\big[V\Xi]_n
	=
	\sum_{k=0}^n
	\binom{n}{k}^{1/2}
	\Symm_n \Xi_{k,n-k}
	\,,
	\label{eqn_isoexpli}
\end{equation}
where for $f: \Rl^n \to \Cl$,
\[
	[\Symm_n f](p_1, \ldots, p_n)
	:=
	\frac{1}{n!} \sum_{\pi \in \frS_n} f(p_{\pi(1)}, \ldots, p_{\pi(n)})
\]
denotes total symmetrization.
Combining this with \eqref{eqn_Saction}, we find
\begin{align*}
	[&\hat{S}_{R} (e^{\Psi_1\oplus\Phi_1})]_n(p_1,\ldots,p_n)
	=
	\sum_{k=0}^n\binom{n}{k}^{1/2}
	\Symm_n(S_{R}(e^{\Psi_1}\otimes e^{\Phi_1})_{k,n-k})(p_1,\ldots,p_n)
	\\
	&=
	\sum_{k=0}^n \frac{\binom{n}{k}^{1/2}}{n!}\sum_{\pi\in\frS_n}
	\prod_{\substack{i=1\ldots k\\j=k+1\ldots n}} \!\!
	 R_0^+(p_{\pi(i)},p_{\pi(j)})	\frac{\Psi_1(p_{\pi(1)})\cdots\Psi_1(p_{\pi(k)})\cdot\Phi_1(p_{\pi(k+1)})\cdots\Phi_1(p_{\pi(n)})}{\sqrt{k!}\sqrt{(n-k)!}}
	\,.
\end{align*}
In the second line, $R_0$ was replaced by $R_0^+$, which does not change the result since the factors of $\Psi_1$ and $\Phi_1$ (explicitly writing out the embedding $\Psi_1 \circ \iota_+$ and $\Phi_1 \circ \iota_-$) are equal to zero unless $p_{\pi(i)} >0$ and $p_{\pi(j)} < 0$, and $R_0(p,q)=R_0^+(p,q)$ for $p>0$, $q<0$.

Next we change the range of indices $i=1,...,k$, $j=k+1,...,n$ of the product to $i=1,...,n$, $j=1,...,n$.
This does not change the result because for the indices $i,j$ which were not present before, we have $R_0^+(p_{\pi(i)},p_{\pi(j)})=1$ on the support of the
remaining factors. After these manipulations, $(p_1,...,p_n)\mapsto\prod_{i,j=1}^n R_0^+(p_{\pi(i)},p_{\pi(j)})$
is a totally symmetric function independent of $\pi$ and $k$. Thus we get
\begin{align*}
	[\hat{S}_R&(e^{\Psi_1\oplus\Phi_1})]_n(p_1,\ldots,p_n)
	\\
	&=
	\prod_{i,j=1}^n R_0^+(p_i,p_j)
	\sum_{k=0}^n\binom{n}{k}^{1/2}\frac{1}{n!}\sum_{\pi\in\frS_n}
	 \frac{\Psi_1(p_{\pi(1)})\cdots\Psi_1(p_{\pi(k)})\cdot\Phi_1(p_{\pi(k+1)})\cdots\Phi_1(p_{\pi(n)})}{\sqrt{k!}\sqrt{(n-k)!}} \\
	&=
	\prod_{i,j=1}^n R_0^+(p_i,p_j)
	\left[ V ( e^{\Psi_1} \otimes e^{\Phi_1}) \right]_n (p_1, \ldots, p_n) \\
	& =
	\prod_{i,j=1}^n R_0^+(p_i,p_j)
	\left[ e^{\Psi_1 \oplus \Phi_1} \right]_n (p_1, \ldots, p_n)
	\,,
\end{align*}
and the proof is finished.
\end{proof}

After these preparations, we can now state the precise relation between the
generators appearing in the two types of deformations.

\begin{proposition}\label{prop:CreatAnniEquiv}
        Let $R\in\R$ be a root of a symmetric inner function. Then, $\psi_\pm \in \Hil_1^\pm$,
        \begin{align}
			a_R(\psi_+)
			=
			V S_R^* (a_+(\psi_+) \otimes 1) S_R V^*
			\,,\qquad
			a_R (\psi_-)
			=
			V S_R (a_+(\psi_-) \otimes 1) S_R^* V^* \,.
			\label{eq:Mainrel}
        \end{align}
\end{proposition}
\begin{proof}
	Let $\Psi\in\Hil$ be a vector of finite particle number. Using
	$a \circ \iota_+ = V (a_+ \otimes 1) V^*$,
	$a \circ \iota_- = V (1 \otimes a_-) V^*$ and the corresponding
	relations for $a^\dagger$ and $a^\dagger_\pm$, we
	can equivalently show $\hat{S}_R^* a(\psi_+) \hat{S}_R = a_R(\psi_+)$,
	$\hat{S}_R a(\psi_-) \hat{S}_R^* = a_R(\psi_-)$. To this end, we compute
        ($p_1,...,p_n\in\Rl$):
	\begin{align}
		[&\hat{S}_R a(\psi)\hat{S}_R^*\Psi]_n(p_1,\ldots,p_n)
		=
		\prod_{i,j=1}^n {R_0^+}(p_i,p_j)
		\cdot
		\sqrt{n+1}
		\int \frac{d q}{\lvert q \rvert}\, \overline{\psi(q)}\,
		[\hat{S}_R^*\Psi]_{n+1}(q,p_1,\ldots,p_n)
		\nonumber
		\\
		&=
		\sqrt{n+1}
		\prod_{i,j=1}^n \!\! R_0^+(p_i,p_j)
		\int \frac{d q}{\lvert q \rvert} \overline{\psi(q)}\,
		\prod_{i',j'=1}^n \!\! \overline{{R_0^+}(p_{i'},p_{j'})}
		\cdot
		\prod_{k=1}^n\overline{{R_0^+}(q,p_k){R_0^+}(p_k,q)}\,
		\Psi_{n+1}(q,p_1,..,p_n)
		\nonumber
		\\
		&=
		\sqrt{n+1} \:
		\bigg\{
		\int_0^\infty \frac{d q}{\lvert q \rvert}\, \overline{\psi(q)}\,
		\prod_{k=1}^n \overline{{R_0^+}(q, p_k)}
		\Psi_{n+1}(q,p_1,\ldots,p_n)
		\label{eq:aR-comparison}
		\\
		&\qquad\qquad\qquad\qquad
		+
		\int_{-\infty}^0 \frac{d q}{\lvert q \rvert}\, \overline{\psi(q)}\,
		\prod_{k=1}^n \overline{{R_0^+}(p_k,q)}
		\Psi_{n+1}(q,p_1,\ldots,p_n)
		\bigg\}
		\,,
		\nonumber
	\end{align}
	where in the last equality $R_0^+(p,q) = 1$ unless
	$p>0$ and $q<0$ was used.

	On the other hand, using $R_0(p,q) = R_0^+(p,q)$ for $p>0, q<0$, and $R_0(p,q) = \overline{R_0^+(q,p)}$ for $p<0, q>0$,
	we find
	\begin{align}
		[a_R(\psi)\Psi]_n(p_1,\ldots,p_n)
		&=
		\sqrt{n+1}
		\int \frac{d q}{\lvert q \rvert}\, \overline{\psi(p)}\,
		\prod_{k=1}^n R_0(q, p_k)
		\Psi_{n+1}(q,p_1,\ldots,p_n)
		\nonumber \\
		&=
		\sqrt{n+1}
		\bigg\{
		\int_0^\infty \frac{d q}{\lvert q \rvert}\, \overline{\psi(q)}\,
		\prod_{k=1}^n R_0^+(q, p_k)
		\Psi_{n+1}(q,p_1,\ldots,p_n)
		\label{eq:Deffullpaxis}
		\\
		&\qquad\qquad\qquad
		+
		\int_{-\infty}^0 \frac{d q}{\lvert q \rvert}\, \overline{\psi(q)}\,
		\prod_{k=1}^n \overline{R_0^+(p_k,q)}
		\Psi_{n+1}(q,p_1,\ldots,p_n)
		\bigg\}
		\,, \nonumber
	\end{align}
	from which we read off $\hat{S}_R a(\psi)\hat{S}_R^*=a_R(\psi)$ for $\supp\psi\subset\Rl_-$.
	For $\supp\psi\subset\Rl_+$, the remaining integrals in \eqref{eq:aR-comparison} and \eqref{eq:Deffullpaxis} agree up to complex conjugation of $R_0^+$; this
	is compensated  by using $\hat{S}_R^*=\hat{S}_{\overline{R}}$ , i.e. in this case we have
	$\hat{S}_R^* a(\psi)\hat{S}_R=a_{R}(\psi)$.
\end{proof}

To obtain the equivalence of Borchers triples, recall that the massless field $\phi_0$ decomposes into chiral components $\phi_{0,\pm}$, each depending on one light ray coordinate $x_\mp=x_0\mp x_1$ only, namely for $f$ which is the derivative of a function in $\Ss(\Rl^2)$,
\begin{align}
	\phi_0(f)
	&=
	V\big(\phi_{0,+}(f_+)\ot1+1\ot\phi_{0,-}(f_-)\big)V^*
	\,,\qquad
	\phi_{0,\pm}(f_\pm)
	=
	\ad_\pm(\fti_\pm|_{\Rl_\pm})+a_\pm(\widetilde{\overline{f_\pm}}|_{\Rl_\pm})
	\,,\nonumber
	\\
	f_ \pm(\mp x_\mp)
	&=
	\frac{1}{2\sqrt{2\pi}}\int_\Rl dx_\pm\,f\left(\tfrac{1}{2}(x_++x_-),\tfrac{1}{2}(x_+-x_-)\right)
	\,.
	\label{eq:fpm}
\end{align}

The algebras in question are generated by these field operators (all of which are essentially self-adjoint on their respective subspaces of finite particle number) by
\begin{align}
	\M_0
	&=
	\{e^{i\phi_0(f)}\,:\,f\in\Ss_\Rl(W)\}''
	\,,\\
	\M_{0,\pm}
	&=
	\{e^{i\phi_{0,\pm}(f_\pm)}\,:\,f\in\Ss_\Rl(W)\}''
	=
	\{e^{i\phi_{0,\pm}(g)}\,:\,g\in\Ss_\Rl(\Rl_+)\}''
	\,.
	\label{eq:M0pm-fields}
\end{align}
We now come to the main result of this section.
\begin{theorem}\label{thm:Algeq}
	Let $R\in\R$. Then $(\M_{R,0},U,\Om)\cong(\NN_R,U_+\ot U_-,\Om_+\ot\Om_-)$.
\end{theorem}
\begin{proof}
	The equivalence $(\NN_R,U_+\ot U_-,\Om_+\ot\Om_-)\cong (\hat{\NN}_R,U_+\ot U_-,\Om_+\ot\Om_-)$ was established already, and by construction of $V$, we have $V(U_+(x_-)\ot U_-(x_+))V^*=U(x)$ and $V\Om_+\ot\Om_-=\Om$. Hence the claim follows once we have shown $V\hat {\NN}_RV^*=\M_{R,0}$.

	By \eqref{eq:Nhat0} and \eqref{eq:M0pm-fields}, $\hat{\NN}_R$ is generated by the bounded functions of the field operators $S_R^*(\phi_{0,+}(f_+)\ot1)S_R$ and $S_R(1\ot\phi_{0,-}(f_-))S_R^*$, $f_\pm\in\Ss(\Rl_+)$. Conjugating with $V$, we have
	\begin{align*}
		VS_R^*\big(\phi_{0,+}(f_+)\ot1\big)S_RV^*
		&=
		\hat{S}_R^*\big(\ad(\fti_+|_{\Rl_+})+a(\widetilde{\overline{f_+}}|_{\Rl_+})\big)\hat{S}_R
		\\
		&=
		\ad_R(\fti_+|_{\Rl_+})+a_R(\widetilde{\overline{f_+}}|_{\Rl_+})
		\,,
	\end{align*}
	where in the second step, we have used Proposition~\ref{prop:CreatAnniEquiv}, which also holds for the creation operators by taking adjoints. Given $f_+$ which is the derivative of a function in $\Ss(\Rl_+)$, we find $f$ which is the derivative in $x_+$ direction of some function in $\Ss(W)$ such that $f_+$ is recovered from $f$ by \eqref{eq:fpm} and $f_-=0$ (namely, one can take the product of $f_+$ (which is a function of $x_-$) and a function of $x_+$ with integral one). In this situation, $f^+=\fti_+|_{\Rl_+}$, $\overline{f^-}=\widetilde{\overline{f_+}}|_{\Rl_+}$, and thus
	\begin{align*}
		VS_R^*\big(\phi_{0,+}(f_+)\ot1\big)S_RV^*
		=
		\phi_{R,0}(f)
		\,.
	\end{align*}
	As all vectors of finite particle number are analytic for these field operators, this equivalence also holds for their associated unitaries $e^{i\phi_{R,0}(f)}$ and $e^{i\phi_{0,+}(f_+)}$, and thus we have the inclusion $VS_R^*(\M_{0,+}\ot 1)S_RV^*\subset\M_{R,0}$.

	Similarly, for the other light ray we obtain
	\begin{align*}
		VS_R\big(1\ot \phi_{0,-}(f_-)\big)S_R^*V^*
		&=
		\ad_R(\fti_-|_{\Rl_-})+a_R(\widetilde{\overline{f_-}}|_{\Rl_-})
		=
		\phi_{R,0}(f)
	\end{align*}
	for suitably chosen $f\in\Ss(W)$, and hence $VS_R(1\ot\M_{0,-})S_R^*V^*\subset\M_{R,0}$. Thus $V\hat{\NN}_RV^*\subset\M_{R,0}$. As $\Om$ is cyclic and separating for both $V\hat{\NN}_R V^*$ and $\M_{R,0}$, and their modular groups w.r.t. $\Om$ coincide, $\Delta_{V\hat{\NN}_RV^*}^{it}=V\Delta_\ot^{it}V^*=\Delta_{\M_{R,0}}^{it}$, the equality of von Neumann algebras $V\hat{\NN}_RV^*=\M_{R,0}$ follows by Takesaki's theorem \cite{Takesaki:2003}
(see \cite[Theorem A.1]{Tanimoto:2011-2} for an explicit application).
\end{proof}

Recall that by construction, $\NN_R$ \eqref{eq:NR0} depends on $R$ only via the symmetric inner function $R^2$, i.e. $\NN_{R_1}=\NN_{R_2}$ if $R_1^2=R_2^2$. By the equivalences
\begin{align*}
	(\NN_R,U_+\ot U_-,\Om_+\ot \Om_-)
	\cong
	(\hat{\NN}_R,U_+\ot U_-,\Om_+\ot \Om_-)
	\cong
	(\M_{R,0},U,\Om)
	\,,
\end{align*}
this also implies $(\M_{R_1,0},U,\Om)\cong(\M_{R_2,0},U,\Om)$ if $R_1^2=R_2^2$.


\section{Structure of massive deformations}\label{section:m>0}

The analysis in the previous section resulted in particular in two equivalence properties
of the massless deformed models: On the one hand, the two deformed Borchers triples
$(\M_{R,0},U,\Om)\cong(\NN_{R,0},U_+\ot U_-,\Om_+\ot \Om_-)$ depend only on
the symmetric inner function $\varphi=R^2$, i.e. choosing a different root of $\varphi$
results in an equivalent model. On the other hand, the deformed and undeformed (chiral)
fields are unitarily equivalent. This equivalence however depends on the light ray,
and thus the triples $(\M_{R,0},U,\Om)$ and $(\M_{1,0},U,\Om)$ are not equivalent for general roots $R\in\R$
(the operator $S_{R^2} = S_\varphi$ appears as the S-matrix, an invariant of Borchers triple \cite{Tanimoto:2011-1, DybalskiTanimoto:2010}).

In this section, we show that the first property also holds in the massive case, whereas the second one only holds in a weaker sense which is specified below.

\subsubsection*{Independence of the choice of root}

We begin with a preparatory lemma.
\begin{lemma}\label{lemma:RootEquivalence}
	Let  $m\geq0$ and $r\in\R$ be a root of the (trivial) symmetric inner function $\varphi(t)=1$. Then the operator
	\begin{align}\label{eq:Yr}
		[Y_r\Psi]_n(p_1,\ldots,p_n)
		:=
		\prod_{1\leq i<j\leq n}
		r_m(p_i,p_j)
		\cdot
		\Psi_n(p_1,\ldots,p_n)
	\end{align}
	is a well-defined unitary on $\Hil$ which commutes with the representation $U$, leaves $\Om$ invariant, and satisfies
	\begin{align}\label{eq:UnitaryEquivalenceOfFields}
		Y_{r}\phi_{R,m}(f)Y_{r}^*\Psi
		=
		\phi_{r\cdot R,m}(f)\Psi
		\,,\qquad
		f\in\Ss(\Rl^2)\,,\;R\in\R\,,\;\Psi\in\DD
		\,.
	\end{align}
\end{lemma}
\begin{proof}
	As $r$ is a root of $1$, it takes only the values $\pm1$ and is in particular real. Hence $r_m(p_j,p_i)=\overline{r_m(p_i,p_j)}=r_m(p_i,p_j)$ is symmetric and thus the product in \eqref{eq:Yr} preserves the totally symmetric subspace of $L^2(\Rl^n)$, and $Y_r$ defines a unitary on the Bose Fock space $\Hil$.

	It is clear that $Y_r$ commutes with translations and leaves $\Om$ invariant. To establish \eqref{eq:UnitaryEquivalenceOfFields}, we first calculate for an annihilation operator $a_{R}(\psi)$, $\psi\in\Hil_1$,
	\begin{align*}
		[Y_r&a_{R}(\psi)Y_r^*\Psi]_n(p_1,\ldots,p_n)
		\\
		&=
		\sqrt{n+1}
		\prod_{i<j}r_m(p_i,p_j)
		\int \frac{dq}{\om_m(q)}\,\overline{\psi(q)}\,
		\prod_{k=1}^n R_m(q,p_k)
		\cdot
		[Y_r^*\Psi]_{n+1}(q,p_1,\ldots,p_n)
		\\
		&=
		\sqrt{n+1}
		\prod_{i<j} |r_m(p_i,p_j)|^2
		\int \frac{dq}{\om_m(q)}\,\overline{\psi(q)}\,
		\prod_{k=1}^n\Big( r_m(q,p_k) R_m(q,p_k)\Big)
		\Psi_{n+1}(q,p_1,\ldots,p_n)
		\\
		&=
		[a_{r\cdot R}(\psi)\Psi]_n(p_1,\ldots,p_n)
		\,.
	\end{align*}
	Thus $Y_r a_{R}(\psi)Y_r^*\Psi=a_{r\cdot R}(\psi)\Psi$, and by taking adjoints, we also find $Y_r\ad_R(\psi)Y_r^*\Psi=\ad_{r\cdot R}(\psi)\Psi$. As $\phi_{R,m}(f)=\ad_{R}(f^+)+a_{R}(\overline{f^-})$, the claimed equivalence \eqref{eq:UnitaryEquivalenceOfFields} follows.
\end{proof}

With this lemma, it is now easy to show that the Borchers triple $(\M_{R,m},U,\Om)$ is independent of the choice of root up to equivalence.

\begin{proposition}\label{proposition:ChoiceOfRootDoesntMatter}
	Let $R_1,R_2\in\R$ be roots of the same symmetric inner function $R_1^2=R_2^2$. Then $(\M_{R_1,m},U,\Om)\cong(\M_{R_2,m},U,\Om)$, $m\geq0$.
\end{proposition}
\begin{proof}
	As $R_1^2=R_2^2$, the function $r(t):=R_1(t) R_2(t)^{-1}$ is a root of $1$ as in  Lemma~\ref{lemma:RootEquivalence}, i.e. we have $Y_r\phi_{R_2,m}(f)Y_{r}^*=\phi_{R_1,m}(f)$ \eqref{eq:UnitaryEquivalenceOfFields} for any $f\in\Ss(\Rl^2)$. But these field operators have the dense subspace $\DD$ of vectors of finite particle number as entire analytic vectors \cite{Lechner:2011}, and $Y_r\DD=\DD$. Hence the equivalence \eqref{eq:UnitaryEquivalenceOfFields} lifts to the unitaries $e^{i\phi_{R_k,m}(f)}$, $k=1,2$, $f\in\Ss_\Rl(\Rl^2)$, and the von Neumann algebras they generate, $Y_r\M_{R_1,m}Y_r^*=\M_{R_2,m}$. Since $Y_r$ also commutes with $U$ and leaves $\Om$ invariant, the claimed equivalence of Borchers triples follows.
\end{proof}

As mentioned in Section 1, this result states that within the class of Borchers triples considered here, the inverse scattering problem for the two-particle S-matrix $R^2$ has a unique solution up to unitary equivalence.
For massless asymptotically complete nets, this uniqueness is known in the stronger form that the wave S-matrix and
the free net give an explicit formula to construct the deformed Borchers triple \cite{Tanimoto:2011-1}.
It is an interesting open problem to find its massive counterpart.

\subsubsection*{Equivalence at fixed momentum}

We now come to the discussion of equivalences between deformed and undeformed field operators. In the massless case, this equivalence can be expressed as, $R\in\R$,
\begin{align*}
	a_R(\xi)
	&=
	\begin{cases}
			\hat{S}_R^*a(\xi)\hat{S}_R & \supp\xi\subset\Rl_+\\
			\hat{S}_Ra(\xi)\hat{S}_R^* & \supp\xi\subset\Rl_-
	\end{cases}
	\,,\qquad
	m=0
	\,.
\end{align*}
For $m>0$, the Lorentz group acts transitively on the upper mass shell, so that there is no invariant distinction between its left and right branch. However, we still have an equivalence of the above form {\em at sharp momentum}. Recall that for $p\in\Rl$, the annihilator $a(p)$ is a well-defined unbounded operator on the dense domain $\DD_0\subset\DD$ of vectors $\Psi\in\DD$ of finite particle number with continuous wave functions $\Psi_n\in C(\Rl^n)$, $n\in\Nl$.

To implement this equivalence, we define an operator-valued function $\Rl\ni p\mapsto\hat{S}_{R,m}(p)\in\U(\Hil)$ by
\begin{align}
	[ \hat{S}_{R,m}(p) \Psi ]_n(p_1,\ldots,p_n)
	:=
	\prod_{1\leq i<j\leq n} R((p_i+p_j)\wedge_m p)
	\cdot
	\Psi_n(p_1, \ldots, p_n)\,,
\end{align}
where $p\wedge_m q:=\frac{1}{2}(\om_m(q)p-\om_m(p)q)$. Note that in case the root $R$ is continuous, one has $\hat{S}_{R,m}(p)\DD_0=\DD_0$.

Using $\overline{R(t)}=R(t)^{-1}=R(-t)$, the definition of $R_m$ \eqref{eq:R_m}, and $(p+q)\wedge_m p=q\wedge_m p$, we then get
\begin{align*}
	[&\hat{S}_{R,m}(p)a(p)\hat{S}_{R,m}(p)^* \Psi]_n (p_1, .., p_n) \\
        &=
	\sqrt{n+1}\,\prod_{i<j}^n  R((p_i+p_j)\!\wedge_m p)
	\cdot
	[\hat{S}_{R,m}(p)^* \Psi]_{n+1}(p,p_1, .., p_n) \\
	&=
	\sqrt{n+1}\,
	\prod_{i<j}^n \Big( R((p_i+p_j)\wedge_m p) \overline{R((p_i+p_j)\wedge_m p)}\Big)
	\cdot
	\prod_{k=1}^n
	\overline{R((p+p_k)\wedge_m p)}
	\cdot
	\Psi_{n+1}(p,p_1, ..., p_n) \\
	&=
	\sqrt{n+1}\,
	\prod_{k=1}^n
	\overline{R(p_k \wedge_m p)}
	\cdot
	\Psi_{n+1}(p,p_1, ..., p_n) \\
	&=
	\sqrt{n+1}\,
	\prod_{k=1}^n
	R_m(p,p_k)
	\cdot
	\Psi_{n+1}(p,p_1, ..., p_n)\\
	&=
	[a_R(p)\Psi]_n(p_1,...,p_n)
	\,,
\end{align*}
where the last equality follows from comparison with \eqref{eq:a_R-explicit}. We thus have on $\DD_0$
\begin{align}
	\hat{S}_{R,m}(p)a(p) \hat{S}_{R,m}(p)^*
	=
	a_{R,m}(p)
	\,.
\end{align}
It should be noted that there is actually a big freedom in the choice of $\hat{S}_{R,m}$
with this property, as it is only the adjoint action of $\hat{S}_{R,m}(p)$ on $a(p)$ with the same 
momentum $p$ that matters in the end. One manifestation of this freedom is the fact that $\hat{S}_{R,0}(p)$
for $p>0$ does not agree with $\hat{S}_R^*$ \eqref{eq:Shat}, whereas their adjoint action on $a(p)$ does.
For $m>0$, another implementation of the equivalence is
\[
(\hat{S}_{R,m}(p)\Psi)_n(p_1,\ldots p_n) =
\prod_{i<j} \overline{R({\rm sgn}(\max(p_j,p_i)-p)
|p_i\wedge_m p_j|)}\cdot\Psi_n(p_1,\ldots, p_n)
\,,
\]
where the sign function sgn is defined with sgn$(0):=-1$. This can be checked by a computation analogous to the previous one. Observe that if $p$ is sufficiently large,
this coincides with the root of the two-particle S-matrix \cite{Lechner:2011}, and for $p$ sufficiently small with its inverse.
Hence this latter implementation is analogous to the massless case, where
the deformation is given exactly by the S-matrix $\hat{S}_R$ and its adjoint.

By (formally) taking adjoints one gets the same relation between $a^\dagger(p)$ and
$a^\dagger_{R,m}(p)$. However, even when making this adjoint rigorous (e.g. in the sense
of quadratic forms) one cannot expect to get an equivalence of the Fourier transform
of the deformed and undeformed field at sharp $p$. One has to keep in mind that the splitting into
chiral components is \emph{not} a splitting of the field according to momentum transfer
but related to a split of the one-particle Hilbert space into positive and negative momentum
parts. Thus creation and annihilation operators appear either both with positive or both
with negative momentum, so both are transformed with $\hat{S}_{R,m}$ or $\hat{S}_{R,m}^*$.
For the massive case this mechanism is not available and therefore the relations between deformed and undeformed creation and annihilation operators will not yield a corresponding relation between the fields.

To conclude, the structure of the wedge algebra is deformed in a very transparent manner in the chiral situation \eqref{eq:NR0}, but not for $m>0$, where one has to rely on the use of generating fields. This observation is to some extent in parallel with the simpler structure of the wave S-matrix in the chiral case in comparison to the many particle S-matrix in the massive case, and deserves further investigation.


\footnotesize{

}

\end{document}